\documentclass[journal]{IEEEtran}

\IEEEoverridecommandlockouts 

\usepackage{graphicx}      

\usepackage{commath}
\usepackage{amsfonts}
\usepackage{mathtools}
\usepackage{steinmetz}		
\usepackage{supertabular}	
\usepackage{tabularx}		
\usepackage{algorithm}		
\usepackage{algorithmic}		
\usepackage{amsmath,amssymb}
\usepackage{amsthm}
\newtheorem{rem}{Remark}
\newtheorem{defi}{Definition}
\newtheorem{prop}{Proposition}

\theoremstyle{proof}

\usepackage{subfig}
\usepackage{enumitem} 		
\usepackage{pgfplots}
\pgfplotsset{width=8cm,height=4cm,compat=1.9}

\usepackage{soul}

\usepackage{cite}

\ifCLASSINFOpdf
\else
\fi
\usepackage{array}
\begin{document}
%
\title{Comparison of Bounds for Optimal PMU Placement for State Estimation in Distribution Grids}


\author{Miguel~Picallo,
		Adolfo~Anta
		and~Bart~De~Schutter,~\IEEEmembership{Fellow,~IEEE}%
\thanks{This project has received funding from the European Union's Horizon 2020 research and innovation programme under the Marie Skł{\l}odowska-Curie grant agreement No 675318 (INCITE).}
\thanks{M. Picallo and B. De Schutter are with the Delft Center for Systems and Control, Delft University of Technology, The Netherlands
{\tt\small \{m.picallocruz,b.deschutter\}@tudelft.nl}}%
\thanks{A. Anta is with the Austrian Institute of Technology 
{\tt\small Adolfo.Anta@ait.ac.at}}%
} 




%


\maketitle

\begin{abstract}
The lack of measurements in distribution grids poses a severe challenge for their monitoring: since there may not be enough sensors to achieve numerical observability, load forecasts (pseudo-measurements) are typically used, and thus an accurate state estimation is not guaranteed. However, an estimation is required to control distribution grids given the increasing amount of distributed generation. Therefore, we consider the problem of optimal sensor placement to improve the state estimation accuracy in large-scale, 3-phase coupled, unbalanced distribution grids. This is a combinatorial optimization problem whose optimal solution is unpractical to obtain for large networks. We explore the properties of different metrics in the context of optimal experimental design, like convexity and modularity, to propose and compare several tight lower and upper bounds on the performance of the optimal solution. Moreover, we show how to use these bounds to choose near-optimal solutions. We test the method on two IEEE benchmark test feeders, the 123-bus and the 8500-node feeders, to show the effectiveness of the approach.
\end{abstract}


\begin{IEEEkeywords}
optimal sensor placement, phasor measurement units, distribution grid state estimation, submodular maximization, projected gradient descent, optimal design of experiments
\end{IEEEkeywords}

%
\IEEEpeerreviewmaketitle

\vspace{-0.02cm}
\section{Introduction}
\vspace{-0.01cm}
The operation of a power network requires accurate monitoring of its state: bus voltages, line currents, consumption and generation, to efficiently manage its controllable elements. State Estimation (SE) serves that purpose by estimating a minimum representation of the network state, for example the bus voltage phasors. Taking several measurements and the admittance matrix of the network as parameters, SE typically solves a weighted least-squares problem using an iterative approach like Newton-Raphson \cite{abur2004power, monticelli2000electric}. In transmission networks, SE is fundamental, since the volatile and distributed generation, which injects power in different locations of the network, causes fast-changing bidirectional power flows. In contrast, SE has not been so necessary in distribution grids until recently, since these grids used to have a simple radial structure with a single source bus injecting power. But this situation is changing due to the increasing penetration of distributed generation like PV panels, batteries, etc. \cite{ipakchi2009grid}.

The introduction of Phasor Measurement Units (PMUs) improves the monitoring of electrical networks. Some work in the literature is focused on using PMUs to achieve topological observability \cite{baldwin1993power}, which can be solved using integer linear programming \cite{gou2008generalized}. Although cheap PMUs \cite{pinte2015mpmu} are becoming available for mass deployment, their operational and network communication costs \cite{beg2016pmucom} may still prevent installing the required minimum number of sensors 
to achieve topological observability based only on PMUs, especially in distribution grids. If topological observability is not possible, neither is numerical observability ensured, which is required to solve the SE \cite{baldwin1993power}. As a result, conventional Supervisory Control And Data Adquisition (SCADA) measurements need to be combined with PMUs to solve the SE problem \cite{zhou2006alternative}. Therefore, authors of recent work \cite{li2011phasor, kekatos2012optimal} formulate the problem of optimal PMU placement in terms of maximising the SE accuracy. This accuracy is typically measured through some metric in the context of optimal design of experiments \cite{pukelsheim2006optimal}, applied to the covariance matrix of the SE error caused by noisy measurements and unreliable forecasts. As a result, the optimal PMU placement problem is a combinatorial optimization problem with a nonlinear objective function, the SE accuracy metric, and thus cannot be solved using linear programming. Moreover, as the size of the network increases, the optimal solution becomes unpractical to obtain given the number of possible combinations of measurements \cite{li2011phasor}.

In distribution grids, one of the major limitations for SE is the lack of sufficient real-time measurements to achieve observability, even using SCADA, and thus SE algorithms need to rely on pseudo-measurements, such as load forecasts. These pseudo-measurements have typically a large relative noise associated (approximately $50$\% \cite{schenato2014bayesian}), which causes a high uncertainty in the SE estimates \cite{picallo2017twostepSE}. As a result, there is a growing interest in using PMUs in distribution grids \cite{meier2014pmudg}, and some recent work proposes algorithms to place sensors to satisfy a desired performance, like algorithms using greedy and random combinations of sensors \cite{singh2009measurement, singh2011meter} or evolutionary algorithms \cite{liu2012tradeoff, Prasad2018tradeoff}. However, these approaches have no optimality guarantees. Since optimal solutions are unpractical in large networks, \cite{fusco2015plac} uses the results in \cite{li2011phasor} to derive lower bounds for the values of the optimal solution.

Our contribution consists in proposing and comparing several bounds for the optimal solution of the PMU placement problem in distribution grids, and extending these bounds to large-scale grids. With these bounds, we will be able to check the gap between any given suboptimal solution and the optimal solution. First, we prove properties like convexity and supermodularity for some metrics in \cite{pukelsheim2006optimal}; then, we use these properties to derive a combination of bounds based on convex optimization as in \cite{kekatos2012optimal}, and bounds based on supermodular minimization similar to \cite{li2011phasor}. These bounds allow us to obtain a close bound for the performance of the optimal solution of the problem under a cardinality constraint, as well as under a budget constraint and sensors with different costs. Additionally, we show how these bounds scale to large networks by using a projected gradient descent algorithm. Moreover, we show through two examples of test feeders, one of them large-scale, how bounds based on the supermodularity property perform better than those based on convexity if the number of sensors to deploy is small, as in distribution networks. 

The rest of the paper is structured as follows. Section \ref{sec:nomen} contains the nomenclature for the most relevant symbols. Section \ref{sec:grid} presents some background about power networks. Section \ref{sec:info} discusses the different types of measurements. Section \ref{sec:SEbasic} summarizes the newly proposed methodology for SE \cite{picallo2017twostepSE}. Section \ref{sec:metrics} presents the metrics considered for the problem and  their properties. Section \ref{sec:optplac} states the optimal sensor placement problem and derives the lower and higher bounds for each optimum value. Section \ref{sec:testcase} shows the effectiveness of the bounds on a test case. Finally, Section \ref{sec:conc} presents the conclusions. For clarity purposes, proofs can be found in the Appendices \ref{sec:appsupermodD} and \ref{sec:appproofbeta}. Appendices \ref{sec:appprojgrad} and \ref{sec:appgreedy} describe the methods to extend the approach to large grids.

\vspace{-0.02cm}
\section{Nomenclature}\label{sec:nomen}
\vspace{-0.02cm}
\setlength{\tabcolsep}{0.1cm}
\hspace{-0.5cm}
\begin{tabular}{ll}
\textbf{variables:} & \\
$V,I,S$ & Vectors of bus voltages, currents and \\ & apparent power\\
$V_\text{src},I_\text{src},S_\text{src}$ & Vectors of bus voltages, currents and \\ & apparent power at the source bus \\
$Y$ & Admittance matrix \\
$N$ & Number of nodes \\
$N_\text{meas}$ & Number of measurements \\
$\sigma_\text{psd}$ & Variance of the pseudo-measurements noise\\ 
$\sigma_\text{mag},\sigma_\text{ang}$ & Variance of the magnitude and angle noise \\ & of the real-time measurements\\ 
$z_\text{meas}, \Sigma_\text{meas}$ & Vector and covariance of measurements \\ 
$C_\text{meas}, \tilde{C}_\text{meas}$ & Matrices mapping state to measurements \\ 
$V_\text{prior}, V_\text{post}$ & Prior and posterior voltage estimation \\ 
$F$ & Subspace of the feasible voltage solutions \\
$\Sigma_\text{prior}, \Sigma_{F,\text{prior}}$ & Covariance of prior estimation \\ 
$\Sigma_\text{post}, \Sigma_{F,\text{post}}$ & Covariance of posterior estimation \\ 
$f_\text{A}(x), f_\text{D}(x)$ & Covariance metrics on the vector space \\
$\tilde{f}_\text{A}(X), \tilde{f}_\text{D}(X)$ & Covariance metrics on the set space \\
$x_\text{opt}, f_\text{opt}$ & Optimal solution and value of the original \\ & PMU allocation problems \eqref{eq:optpmu}, \eqref{eq:optpmubudget} \\
$x_\text{convex}, f_\text{convex}$ & Optimal solution and value of the problems \\ & with convex relaxations \eqref{eq:optpmuconvex}, \eqref{eq:optpmuconvexbudget} \\
$x_\text{feas}, f_\text{feas}$ & Feasible solution and value of \eqref{eq:optpmufeas}, \eqref{eq:optpmufeasbudget} \\
$x_\text{greedy}, f_\text{greedy}$ & Greedy solution and value of \eqref{eq:optpmugreedy} and \\ & Algorithm $\ref{alg:greedybudget}$ \\

\end{tabular}

\hspace{-0.5cm}
\begin{tabular}{ll}
\textbf{symbols:} & \\
$\abs{\cdot}$ & Magnitude of a complex number \\
$\bar{(.)},(\cdot)^*$ & Complex conjugate and conjugate transpose \\
diag$(\cdot)$ & Diagonal operator: vector to diagonal matrix \\
$(\cdot)_i$ & Element $i$ of a vector \\
$(\cdot)_{i,j}$ & Element in row $i$ and column $j$ of a matrix \\
$(\cdot)_{\varepsilon}$ & Elements of a vector at indices in $\varepsilon$\\
$(\cdot)_{j,\bullet}, (\cdot)_{\varepsilon,\bullet}$ & Row $j$ or rows with indices $\varepsilon$ of matrix \\
tr$(\cdot)$,det$(\cdot)$ & Trace and determinant of a matrix\\
$\nabla f$ & Gradient of $f$ \\
$\{\cdot\}$ & Set of elements \\
$(\cdot)^{(K)}$ & Value at iteration $K$ \\
$\Pi_\mathcal{X}(\cdot)$ & Projection onto the set $\mathcal{X}$ \\
\end{tabular}

\vspace{-0.02cm}
\section{State Estimation in Distribution Grids}\label{sec:SEingrid}
\vspace{-0.02cm}
\subsection{Distribution Grid Model}\label{sec:grid}
\vspace{-0.01cm}
A distribution grid consists of buses, where power is injected or consumed, and branches, each connecting two buses. This system can be modeled as a graph $\mathcal{G}=(\mathcal{V},\mathcal{E},\mathcal{W})$ with nodes $\mathcal{V}=\{1,...,N_\text{bus}\}$ representing the buses, edges $\mathcal{E}=\{(v_i,v_j)\mid v_i,v_j \in \mathcal{V}\}$ representing the branches, and edge weights $\mathcal{W}=\{w_{i,j}\mid (v_i,v_j) \in \mathcal{E}\}$ representing the admittance of the branches, which are determined by the length and type of the line cables.
In 3-phase networks buses may have up to 3 phases, so that the voltage at bus $i$ is $V_{\text{bus},i} \in \mathbb{C}^{n_{\phi,i}}$, where $n_{\phi,i}\leq 3$ (and the edge weights $w_{i,j}\in \mathbb{C}^{n_{\phi,i} \times n_{\phi,j}}$). The state of the network is then typically represented by the vector bus voltages $V_\text{bus}=[V_\text{src}^T, \; V^T]^T \in \mathbb{C}^{N+3}$, where $V_{\text{src}} \in \mathbb{C}^3$ denotes the known voltage of the 3 phases at the source bus, and $V \in \mathbb{C}^N$ the voltages in the non-source buses, where $N$ depends on the number of buses and phases per bus. Then, using the Laplacian matrix $Y \in \mathbb{C}^{(N+3) \times (N+3)}$ of the weighted graph $\mathcal{G}$, called admittance matrix \cite{abur2004power}, the power flow equations to compute the currents $I$ and the power loads $S$ are:
\begin{equation}\label{eq:PFeq}\arraycolsep=1pt
\begin{array}{c}
\left[\begin{array}{c} I_{\text{src}} \\ I \end{array}\right] = 
Y\left[\begin{array}{c} V_{\text{src}} \\ V \end{array}\right], \; S = \text{diag}(\bar{I})V
\end{array}
\end{equation}

\vspace{-0.02cm}
\subsection{Measurements} \label{sec:info}
\vspace{-0.01cm}
As explained in \cite{picallo2017twostepSE}, several different sources of information can be available to solve the SE problem:
\begin{enumerate}[leftmargin=*]
\item \textit{Pseudo-measurements}, i.e., load estimations $S_\text{psd}$ based on predictions and/or known installed load capacity at every bus. Since these pseudo-measurements are estimations rather than actual measurements, we model their uncertainty using a Gaussian noise with a relative large standard deviation (a typical value can be $\sigma_\text{psd} \approx 50\%$ \cite{schenato2014bayesian}).

\item \textit{Virtual measurements}, i.e., buses with zero-injections, no loads connected. They can be represented as physical constraints for the voltage states by defining the set of indices of zero-injection buses $\varepsilon=\{i,\cdots,j\}$:
\begin{equation}\label{eq:Scons}
(S)_{\varepsilon}=0, \; (I)_{\varepsilon}=0
\end{equation} 

\item \textit{Real-time PMU measurements}, i.e., voltage and current GPS-synchronized measurements of magnitude and phase angle. According to the IEEE standard for PMU \cite{martin2008exploring}, these measurements may have a small error. Again, we model this uncertainty using a Gaussian noise with a low standard deviation for the magnitude and the angle, $\sigma_\text{mag} \approx 1\%$ and $\sigma_\text{ang} \approx 0.01\text{ rad}$ respectively. They can be expressed using a linear approximation \cite{picallo2017twostepSE} with magnitude and angle noise due to the measurements and imperfect synchronization. For a number of $N_\text{meas}$ measurements $z_\text{meas} \in \mathbb{C}^{N_\text{meas}}$ we have
\begin{equation}\label{eq:Lmeas}
\begin{array}{c}
z_\text{meas} \approx  C_\text{meas}V + \text{diag}(C_\text{meas}V)(\omega_{\text{mag}} + j\omega_{\text{ang}})
\end{array}
\end{equation}
where $\omega_\text{mag}, \omega_\text{ang}$ are the Gaussian noises with mean $0$ and standard deviation $\sigma_\text{mag}, \sigma_\text{ang}$ respectively: $\omega_\text{mag} \hspace{-0.1cm} \sim \hspace{-0.1cm} \mathcal{N}(0,\sigma_\text{mag}I_{\text{d},N_\text{meas}})$, $\omega_\text{ang} \hspace{-0.1cm} \sim \hspace{-0.1cm} \mathcal{N}(0,\sigma_\text{ang}I_{\text{d},N_\text{meas}})$, with $I_{\text{d},n}$ denoting the identity matrix of dimension $n$, and where the matrix $C_\text{meas}$ maps state voltages to measurements; it relates the values of the measurements to values of state voltages. Then, for measurement $j$ at phase $l$ of bus $i$ we have
\begin{equation}\label{eq:LmeasMap}\arraycolsep=1pt\begin{array}{l}
(C_\text{meas}V)_j = (C_\text{meas})_{j,\bullet}V= \\[0.1cm]
\left\lbrace \begin{array}{ll}
V_{i_l} & \mbox{for a voltage measurement} \\[0.0cm]
(Y)_{i_l,\bullet}V & \mbox{for a current measurement} \\[0.0cm]
(Y)_{i_l,m_l}(V_{i_l}-V_{m_l}) & \mbox{for a branch-current $i \to m$} \\[-0.0cm]
&  \mbox{measurement}
\end{array} \right.
\end{array}
\end{equation}
Since the measurement noises in \eqref{eq:Lmeas} are small according to the PMUs standard \cite{martin2008exploring}, their covariance matrices can be approximated using the measurements:
\begin{equation*}
\begin{array}{rl}
\Sigma_\text{meas} & = (\sigma_\text{mag}^2+\sigma_\text{ang}^2) \text{diag}(\:\abs{C_\text{meas}V}^2) \\
& \approx (\sigma_\text{mag}^2+\sigma_\text{ang}^2) \text{diag}(\:\abs{z_\text{meas}}^2)
\end{array}
\end{equation*}
\end{enumerate}
\vspace{-0.02cm}
\subsection{State Estimation}\label{sec:SEbasic}
\vspace{-0.01cm}
Typically, SE consists in finding the voltages that best match the measurements by solving a weighted least-squares problem \cite{abur2004power}. As proposed in \cite{picallo2017twostepSE}, SE can be decomposed in two parts: First, using the pseudo-measurement estimations or predictions for the loads $S_\text{psd}$, we solve the power flow  offline to obtain a prior estimate $V_\text{prior}$: 
\begin{equation}\label{eq:PF}\begin{array}{l}
V_\text{prior} = \text{PowerFlow}(S_\text{psd})
\end{array}
\end{equation}
Then, using the real-time PMU measurements $z_\text{meas}$, a posterior solution $V_\text{post}$ can be derived using a linear filter:
\begin{equation}\label{eq:update}\begin{array}{l}
V_\text{post} = V_\text{prior} + K(z_\text{meas}-C_\text{meas}V_\text{prior}) 
\end{array}
\end{equation}
where the gain matrix $K$ is obtained by minimizing the error covariance $\Sigma_\text{post}=\mathbb{E}[(V_\text{post}-V)^*(V_\text{post}-V)]$:
\begin{equation}\label{eq:sigmapost}
\arraycolsep=1pt
\begin{array}{rl}
\Sigma_\text{post} =&  \Sigma_\text{prior} + K(\Sigma_\text{meas}+C_\text{meas}\Sigma_\text{prior}C_\text{meas}^*)K^* \\[0.1cm]
&-KC_\text{meas}\Sigma_\text{prior}-\Sigma_\text{prior}C_\text{meas}^*K^* 
\\[0.1cm]
K = & \arg\min_K \text{tr}(\Sigma_\text{post}) \\[0.1cm]
= & \Sigma_\text{prior}C_\text{meas}^*(C_\text{meas}\Sigma_\text{prior}C_\text{meas}^*+\Sigma_\text{meas})^{-1}
\end{array}
\end{equation}
where $\Sigma_\text{prior}$ and $\Sigma_\text{meas}$ are the expected error covariance of the prior estimate $V_\text{prior}$ and the measurements $z_\text{meas}$ respectively. 

\begin{rem}
Extra non-synchronized real-time measurements, like magnitude measurements from SCADA, can be included in the posterior update \eqref{eq:update} for a greater improvement of the posterior estimate $V_\text{post}$ by using the first-order approximation of the measurement function \cite{picallo2017twostepSE}. If a multi-stage PMU deployment is considered \cite{aminafar2011multistage}, prior installed PMUs can be also considered as extra measurements in the same manner.
\end{rem}

\begin{rem}
As shown in \cite{zhou2006alternative}, splitting the problem in two steps yields the same first-order approximation as solving the problem in one step. Moreover, the posterior minimum-variance estimator using the linear update \eqref{eq:update}, is equal to the maximum-likelihood using a weighted least-squares approach \cite{picallo2017twostepSE}. Therefore, we can conclude that for an SE method that assumes Gaussian noises and performs a maximum likelihood estimation, the posterior covariance will be approximately the one in \eqref{eq:sigmapost}, and thus the method developed here for optimal sensor placement can be also extended for other SE techniques satisfying these conditions.
\end{rem}

Since $V_\text{post}$ in \eqref{eq:update} is an unbiased estimator, SE accuracy can be defined as a function of the posterior covariance matrix $\Sigma_\text{post}$ in \eqref{eq:sigmapost}, which needs to be minimized to improve the SE accuracy. This motivates a deeper analysis of $\Sigma_\text{post}$: if there are zero-injection buses at indices $\varepsilon$ when solving \eqref{eq:PF}, then $\Sigma_\text{prior}$ will not be of full rank and thus will not be invertible. For convenience, we consider the restricted subspace $\{V\mid (Y)_{\varepsilon,\bullet}V=0\}$ and the linear transformation proposed in \cite{picallo2017twostepSE} to represent the space of feasible solutions: $V=Fx+V_0$ with $x \in \mathbb{C}^{N-\abs{\varepsilon}}$, where $\abs{\varepsilon}$ is the cardinality of $\varepsilon$ and $F$ is the null space of $(Y)_{\varepsilon,\bullet}$: $F=\text{ker}((Y)_{\varepsilon,\bullet})\in \mathbb{C}^{N\times N-\abs{\varepsilon}}$, so that $F^*F=I_\text{d}$, and $V_0$ denotes the voltage under zero loads. Then we have $\Sigma_\text{prior}=F\Sigma_{F,\text{prior}}F^*$, where $\Sigma_{F,\text{prior}}$ is the covariance of $x$. After some manipulations of \eqref{eq:sigmapost}, the resulting error covariance $\Sigma_\text{post}$ ($\Sigma_{F,\text{post}}$ in this subspace) for the posterior estimation $V_\text{post}$ can be expressed as
\begin{equation}\label{eq:sigmaPost2}\arraycolsep=1pt\begin{array}{c}
\Sigma_\text{post}=F(\Sigma_{F,\text{prior}}^{-1}+(C_\text{meas}F)^*\Sigma_\text{meas}^{-1}C_\text{meas}F)^{-1}F^* \\[0.1cm]
\Sigma_{F,\text{post}}=(\Sigma_{F,\text{prior}}^{-1}+(C_\text{meas}F)^*\Sigma_\text{meas}^{-1}C_\text{meas}F)^{-1}
\end{array}
\end{equation}

In the presence of zero-injection buses, $\abs{\varepsilon}>0$, and then $\Sigma_\text{post}\in \mathbb{C}^{N\times N}$ is not full rank: $\text{rank}(\Sigma_\text{post})=N-\abs{\varepsilon}$. It has the same rank as $\Sigma_{F,\text{post}}\in\mathbb{C}^{(N-\abs{\varepsilon}) \times (N-\abs{\varepsilon})}$, which is full rank since $\Sigma_{F,\text{prior}}$ is full rank. Moreover, the eigenvalues of $\Sigma_{F,\text{post}}$ are all eigenvalues of $\Sigma_{\text{post}}$, because for any eigenvector $v$ of $\Sigma_{F,\text{post}}$, $Fv$ is an eigenvector of $\Sigma_{\text{post}}$ with the same eigenvalue, since $F^*F=I_\text{d}$. Consequently, we can analyze $\Sigma_{F,\text{post}}$ instead of $\Sigma_{\text{post}}$. These eigenvalues represent the lengths of the axes of the confidence ellipsoid \cite{boyd2004convex}. Moreover, since measurement errors are caused separately by each sensor, they are independent, i.e. $\Sigma_\text{meas}$ is diagonal, and we can split $\Sigma_{F,\text{post}}$ by every measurement:
\begin{equation}\label{eq:sigmaPostfx}\arraycolsep=1pt\begin{array}{l}
\Sigma_{F,\text{post}}
=(\Sigma_{F,\text{prior}}^{-1}+\sum_i (C_\text{meas}F)_{i,\bullet}^*(C_\text{meas}F)_{i,\bullet}(\Sigma_\text{meas}^{-1})_{i,i})^{-1}\\[0.1cm]
=(\Sigma_{F,\text{prior}}^{-1}+\sum_i x_i(\tilde{C}_\text{meas}F)_{i,\bullet}^*(\tilde{C}_\text{meas}F)_{i,\bullet}(\Sigma_\text{meas}^{-1})_{i,i})^{-1}\\[0.1cm]
= \Sigma_{F,\text{post}}(x)
\end{array}
\end{equation}
where $x_i\in \{0,1\}$, $x_i=1$ if the physical quantity $i$ has a sensor measuring its value, $0$ otherwise; and $\tilde{C}_\text{meas}$ is the special case of  $C_\text{meas}F$ with all possible measurements of all types (bus voltage, bus current, and line current) for all nodes and lines in each phase.

\begin{defi}
In order to improve the accuracy of the SE, the problem of optimal sensor placement consists in minimizing $\Sigma_{\text{post}}(x)$, equivalently $\Sigma_{F,\text{post}}(x)$ in \eqref{eq:sigmaPostfx},
according to a metric $m(\cdot)$ and under a set of constraints $h(\cdot)$ to limit the number of sensors or the total cost:
\begin{equation}\label{eq:optprob}
\min_x m(\Sigma_{F,\text{post}}(x)) \text{ s.t. } h(x)\leq 0, \; x_i \in \{0,1\} \; \forall i
\end{equation}
For simplicity, we define: $ f(x)\equiv m(\Sigma_{F,\text{post}}(x))$.
\end{defi}

\begin{rem}
Problem \eqref{eq:optprob} is equivalent to minimizing the number of sensors or the total cost of the sensors, while enforcing a given accuracy, i.e., satisfying a performance threshold $\lambda$ of the metric: $\min_{\{x\mid f(x)\leq \lambda, \; x_i \in \{0,1\} \; \forall i \}} h(x)$, as in \cite{singh2009measurement, singh2011meter, liu2012tradeoff, Prasad2018tradeoff}. To achieve that, the problem \eqref{eq:optprob} can be solved for different numbers of sensors or budgets until the performance value is below the threshold: $\min_{\{x\mid h(x)\leq 0, \; x_i \in \{0,1\} \; \forall i \}} f(x) \leq \lambda$. The representation in \eqref{eq:optprob}, as used in \cite{li2011phasor, kekatos2012optimal, fusco2015plac} is preferred due to the availability of a computationally efficient optimization algorithm to solve the problem, see Appendix \ref{sec:appprojgrad}.
\end{rem}

\vspace{-0.02cm}
\section{Metrics for Sensor Placement}\label{sec:metrics}
\vspace{-0.01cm}
There are many possible metrics $f(x)$ available in the context of optimal design of experiments \cite{pukelsheim2006optimal}. Concretely, we will focus on the A-\textit{optimal} and the D-\textit{optimal} metrics, because of the properties that we will show later. Nonetheless, the results of this paper could be generalized to other metrics with the same properties.
\begin{defi}
The A-\textit{optimal} and the D-\textit{optimal} metrics are defined as follows:
\begin{equation}\begin{array}{l}
\textnormal{A-}optimal: f_\textnormal{A}(x):=\textnormal{tr}(\Sigma_{\textnormal{post}}(x))=\textnormal{tr}(\Sigma_{F,\textnormal{post}}(x)) \\[0.1cm]
\textnormal{D-}optimal: f_\textnormal{D}(x):=\log(\textnormal{det}(\Sigma_{F,\textnormal{post}}(x)))
\end{array}
\end{equation}
\end{defi}
The A-\textit{optimal} metric represents the sum of the eigenvalues of $\Sigma_{F,\text{post}}$ and thus the sum of the lengths of axes of the confidence ellipsoid \cite{boyd2004convex}. This is the metric typically used for SE methods, since standard SE maximizes the log-likelihood through a weighted least-squares minimization \cite{abur2004power}, which is equivalent to the minimum-variance estimator using the trace \cite{picallo2017twostepSE}. The D-\textit{optimal} metric is the natural logarithm of the product of the eigenvalues of $\Sigma_{F,\text{post}}$ and it is related to the logarithm of the volume of the confidence ellipsoid \cite{boyd2004convex}. These metrics have several properties relevant to the problem of optimal sensor placement:

\begin{enumerate}[leftmargin=*]
\item \textit{Convexity}:
When relaxing $x_i$ in \eqref{eq:optprob} to be continuous, $x_i \in [0,1]$, the metrics are convex on $x$ \cite[Section~7.5]{boyd2004convex} \cite{kekatos2012optimal}, and thus a global optimum can be computed efficiently. Note that despite $\log(x)$ being concave for $x \in \mathbb{R}$, $\log(\text{det}(X))$ is convex for $X \in \mathbb{R}^{n \times n}$ \cite{boyd2004convex}.

\item \textit{Gradient computation}:
The gradients $\nabla f$ for $f \in \{f_\text{A},f_\text{D}\}$ can be derived analytically using matrix calculus \cite{matrixcalc}:
\begin{equation}\label{eq:grad}\arraycolsep=1pt\begin{array}{rl}
(\nabla f_\text{A}(x))_i = & -\text{tr}\Big(\Sigma_{F,\text{post}}^2(x)
(\tilde{C}_\text{meas})_{i,\bullet}^*(\tilde{C}_\text{meas})_{i,\bullet}\Big) (\Sigma_\text{meas}^{-1})_{i,i} \\[0.1cm]

(\nabla f_\text{D}(x))_i = & -\text{tr}\Big(\Sigma_{F,\text{post}}(x)
(\tilde{C}_\text{meas})_{i,\bullet}^*(\tilde{C}_\text{meas})_{i,\bullet}\Big) (\Sigma_\text{meas}^{-1})_{i,i} 
\end{array}
\end{equation}
These expressions will be necessary when developing gradient methods to optimize the placement of sensors in large networks, see Appendix \ref{sec:appprojgrad}.
 
\item \textit{Monotonicity}:
Moreover, these metrics can be seen as set functions by defining the sets $X=\{i\mid x_i=1\}$. Note that there is a bijection between both, i.e., $i\in X \iff x_i=1$. Therefore, we will use the notation $X$ and $x$ to denote the set and the vector respectively, and $\hat{f}(X)=f(x)$ to denote the metric functions applied to them. Defining
\begin{equation*}
\Lambda_X :=  \sum_{i\in X}(\tilde{C}_\text{meas})_{i,\bullet}^*(\tilde{C}_\text{meas})_{i,\bullet}(\Sigma_\text{meas}^{-1})_{i,i} \succeq 0
\end{equation*}
the posterior error covariance can be expressed as $\Sigma_{F,\text{post}}(X)= (\Sigma_{F,\text{prior}}^{-1}+\Lambda_X)^{-1}$.
\begin{prop}\label{prop:monotone}
$\hat{f}_\mathrm{A}(X),\hat{f}_\mathrm{D}(X)$ are monotone decreasing.
\end{prop}
\begin{proof}
Consider the sets of sensors $X\subseteq Y$. Since $\Lambda_Y-\Lambda_X=\Lambda_{Y\setminus X}\succeq 0$, we have $\Lambda_Y\succeq\Lambda_X$. Then $\Sigma_{F,\text{prior}}^{-1}+\Lambda_Y\succeq \Sigma_{F,\text{prior}}^{-1}+\Lambda_X \succeq 0$ and thus $(\Sigma_{F,\text{prior}}^{-1}+\Lambda_X)^{-1}\succeq (\Sigma_{F,\text{prior}}^{-1}+\Lambda_Y)^{-1}$. Consequently, for any of the functions $\hat{f}\in \{\hat{f}_\text{A},\hat{f}_\text{D}\}$, $\hat{f}(X)\geq \hat{f}(Y)$.
\end{proof}

\item \textit{Modularity}: 
Submodularity/supermodularity are a sort of concavity/convexity properties when considering set functions. The increment of a set function when adding a new element diminishes (submodular) or increases (supermodular) as the set gets larger.
\begin{defi}
Given a finite set $\Omega$, a submodular/supermodular function \cite{nemhauser1978analysis} is a set function $\hat{f}:2^\Omega \to \mathbb{R}$, with $\hat{f}(\emptyset)=0$, where $\emptyset$ is the empty set, so that for any element $a \in \Omega$ and two subsets $X,Y$ so that $X \subseteq Y \subseteq \Omega \setminus \{a\}$ we have
\begin{equation*}\arraycolsep=1pt
\begin{array}{ll}
\text{submodular:} & \hat{f}(Y \cup \{a\})-\hat{f}(Y) \leq \hat{f}(X \cup \{a\})-\hat{f}(X) \\[0.1cm]
\text{supermodular:} & \hat{f}(Y \cup \{a\})-\hat{f}(Y) \geq \hat{f}(X \cup \{a\})-\hat{f}(X) 
\end{array}
\end{equation*}
\end{defi}
\begin{prop}\label{prop:supermodular}
$\hat{f}_\mathrm{D}(X)-\hat{f}_\mathrm{D}(\emptyset)$ is supermodular.
\end{prop}
\begin{proof}
See Appendix \ref{sec:appsupermodD}
\end{proof}
This means that for $f_\text{D}$, there exist bounds between greedy solutions and optimal solutions, as we will see later in \eqref{eq:boundsubmod} and \eqref{eq:boundsubmodbudget}. These bounds will allow us to provide a limit on the values of the optimal solution, and thus to check how far the values of other suboptimal solutions can be from the value of the optimal solution.

\begin{rem}
In \cite{li2011phasor} was proven that if the rows in $\tilde{C}_\text{meas}$ are orthogonal, $\hat{f}_\mathrm{A}(X)-\hat{f}_\mathrm{A}(\emptyset)$ is supermodular, and lower bounds based on the supermodularity property can also be used for the A-optimal metric. However, the rows of $\tilde{C}_\text{meas}$ are in general not orthogonal due to the different types of sensors in \eqref{eq:LmeasMap}, the electrical connections represented in $Y$, and the effect of the reduced subspace $F$.
\end{rem}
\end{enumerate}

\vspace{-0.02cm}
\section{Optimal Sensor Placement}
\vspace{-0.01cm}
\label{sec:optplac}
So far we have looked at the cost function in \eqref{eq:optprob}, now we will focus on the constraints. In this section we state the optimal sensor placement problem under two different constraints: a cardinality constraint limiting the number of sensors, which is the most typical approach in the literature \cite{li2011phasor, kekatos2012optimal, schenato2014bayesian}, and a budget constraint limiting the total cost of the deployed sensors, which may have different costs. For each case, we derive the respective lower and upper bounds on the optimal performance, using the metrics properties derived in Section \ref{sec:metrics}.

\vspace{-0.02cm}
\subsection{Solutions with cardinality constraint}
For any of the metrics $f \in \{f_\text{A},f_\text{D}\}$, the problem of optimal placement of $N_\text{meas}$ sensors is
\begin{equation}\label{eq:optpmu}
x_\text{opt} = \arg\min_x f(x) \text{ s.t. } \sum_i x_i \leq N_\text{meas},\; x_i\in \{0,1\}\; \forall i
\end{equation}
where $N_\text{meas}$ indicated the maximum number of sensors. We denote the value of the optimal solution of \eqref{eq:optpmu} $f_\text{opt} = f(x_\text{opt})$.

Similar to \cite{joshi2009sensor,chepuri2015sparsity} for general estimation problems and in \cite{kekatos2012optimal} for transmission power networks, we relax the constraints in \eqref{eq:optpmu}, $x_i \in [0,1]$, to get a continuous convex problem using
\begin{equation}\label{eq:optpmuconvex}
x_\text{convex} = \arg\min_x f(x) \text{ s.t. } \sum_i x_i\leq N_\text{meas},\; x_i\in  [0,1]\; \forall i
\end{equation}
We denote the value of the optimal solution of \eqref{eq:optpmuconvex} as $f(x_\text{convex})=f_\text{convex}$. However, this solution will not necessarily be feasible to \eqref{eq:optpmu}. To create a feasible solution $x_\text{feas}$ (with respective value $f(x_\text{feas})=f_\text{feas}$), we can take the largest $N_\text{meas}$ values of $x_\text{convex}$, set them to $1$ and the others to $0$:
\begin{equation}\label{eq:optpmufeas}
x_{\text{feas},i} = \left\lbrace\begin{array}{l}
1 \mbox{ if }x_{\text{convex},i} \geq x_{\text{convex},k} \\[0.05cm] 
0 \mbox{ otherwise} 
\end{array}\right.
\end{equation}
where $k$ is such that $\abs{ \{ i\mid x_{\text{convex},i}\geq x_{\text{convex},k}\} }=N_\text{meas}$. Ties are broken arbitrarily if there is more than one $x_{\text{convex},i}$ with value $x_{\text{convex},k}$, i.e. if $\abs{ \{ i\mid x_{\text{convex},i}= x_{\text{convex},k}\} } >1$. This also applies for further possible ties throughout the paper. Another simple way to create a feasible solution would be using a forward greedy sensor selection: at iteration $K$, given the set of selected sensors $X^{(K-1)}$, add a new sensor $k$ such that
\begin{equation}\label{eq:optpmugreedy}
k=\arg\min_{i \notin X^{(K-1)}} \hat{f}(X^{(K-1)} \cup \{i\}), \; X^{(K)} = X^{(K-1)} \cup \{k\} 
\end{equation}
We denote the solution of \eqref{eq:optpmugreedy} $x_\text{greedy}$ and its value $f(x_\text{greedy})=f_\text{greedy}$. Then, since $x_\text{opt}$ is a feasible suboptimal solution of \eqref{eq:optpmuconvex}, and $x_\text{feas}$ and $x_\text{greedy}$ are feasible suboptimal solutions of \eqref{eq:optpmu}, the following holds for all metrics:
\begin{equation}\label{eq:bound}
f_\text{convex}\leq f_\text{opt} \leq \min(f_\text{greedy},f_\text{feas}) 
\end{equation} 

For the case of $\hat{f}_\text{D}$, since $\hat{f}_\text{D}(X)-\hat{f}_\text{D}(\emptyset)$ is monotone nonincreasing and supermodular ($-\hat{f}_\text{D}(X)+\hat{f}_\text{D}(\emptyset)$ is monotone nondecreasing submodular), we have an extra lower bound \cite{nemhauser1978analysis} for $f_\text{opt}$ of \eqref{eq:optpmu}: 
\begin{equation}\arraycolsep=1pt \begin{array}{c}
\hat{f}_{\text{D},\text{greedy}}-\hat{f}_\text{D}(\emptyset) \leq (\hat{f}_{\text{D},\text{opt}}-\hat{f}_\text{D}(\emptyset))\alpha  \\[0.1cm]
\alpha = \left(1-\left(1-\frac{1}{N_\text{meas}}\right)^{N_\text{meas}}\right) \in (1-e^{-1},1]
\end{array}
\end{equation}
so that 
\begin{equation}\label{eq:boundsubmod}
\tilde{f}_{\text{D},\text{greedy}} \coloneqq (\hat{f}_{\text{D},\text{greedy}}-\hat{f}_\text{D}(\emptyset))\alpha^{-1}+\hat{f}_\text{D}(\emptyset) \leq f_{\text{D},\text{opt}}
\end{equation}

\vspace{-0.02cm}
\subsection{Solutions with budget constraint}
Although the problem of sensor placement under a cardinality constraint is simpler and has more desired properties \cite{nemhauser1978analysis}, a more realistic and economic representation of the problem would be having a budget constraint on the sensors 
\begin{equation}\label{eq:optpmubudget}
x_\text{opt} = \arg\min_x f(x) \text{ s.t. } \sum_i c_i x_i  \leq b,\; x_i\in\{0,1\} \; \forall i
\end{equation}
where $b$ represents the budget and $c_i$ the cost of installing a sensor at location $i$. This can take into account the extra cost of installing a sensor in a remote area, or locations where specific rights might be required, or different types of sensors, etc. The relaxed convex problem is then
\begin{equation}\label{eq:optpmuconvexbudget}
x_\text{convex} = \arg\min_x f(x) \; s.t. \; \sum_i c_i x_i \leq b,\; x_i\in  [0,1]\; \forall i
\end{equation}

Similar to \eqref{eq:optpmufeas}, a feasible solution $x_\text{feas}$ w.r.t. \eqref{eq:optpmubudget} can be built using the convex solution $x_\text{convex}$ of \eqref{eq:optpmuconvexbudget}, by iteratively taking the sensor with highest $x_{\text{convex},i}$: at iteration $K$, and until the set of possible sensors is empty, $\mathcal{B}=\{i\mid c_i \leq b - \sum_{j\in X^{(K-1)}}c_j \}=\emptyset$, add to the set of selected sensors $X^{(K-1)}$ a new sensor $k \notin X^{(K-1)}$ such that:

\begin{equation}\label{eq:optpmufeasbudget}
k={\arg\max}_{i \in \mathcal{B}, i \notin X^{(K-1)}} x_{\text{convex},i},\; X^{(K)} = X^{(K-1)} \cup \{k\} 
\end{equation}

Likewise, we can adapt the forward greedy selection algorithm \eqref{eq:optpmugreedy} to take costs into account. Therefore, we can use the the cost-effective forward greedy selection algorithm of \cite{leskovec2007cost}, which in every iteration adds the sensor with the lowest ratio of objective function improvement divided by sensor cost. In this version we also return intermediate values required for the bounds. For clarity we detail the whole algorithm here:

\begin{algorithm}[tbh]\label{alg:1}
\caption{Cost-effective forward greedy selection}\label{alg:greedybudget}
\begin{algorithmic}[1]
\REQUIRE $b>0,c\in\mathbb{R}^n,c \geq 0$
\STATE $X_1 \leftarrow \emptyset$
\STATE $k \leftarrow \arg\min_i \frac{\hat{f}(\{i\})}{c_i}$
\WHILE{$c_k\leq b-\sum_{i\in X_1}c_i$}
	\STATE $X_1 \leftarrow X_1\cup \{k\}$
	\STATE $k \leftarrow \arg\min_i \frac{\hat{f}(X_1\cup\{i\})}{c_i}$
\ENDWHILE
\STATE $a \leftarrow k$
\STATE $X_2 \leftarrow X_1$
\WHILE{$ \{i|c_i\leq b-\sum_{j\in X_2}c_j \} \neq \emptyset$}
	\STATE $k \leftarrow \arg\min_{\{i|c_i\leq b-\sum_{j\in X_2}c_j \}} \frac{\hat{f}(X_2\cup\{i\})}{c_i}$
	\STATE $X_2 \leftarrow X_2\cup \{k\}$
\ENDWHILE
\RETURN  $X_1,X_2,a$
\end{algorithmic}
\end{algorithm}

Using Algorithm \ref{alg:greedybudget}, we denote the solutions as $X_\text{greedy1}=X_1$, $X_\text{greedy2}=X_2$ and $X_\text{greedy1a}=X_1\cup\{a\}$  with corresponding values $\hat{f}(X_\text{greedy1})=f_\text{greedy1},\hat{f}(X_\text{greedy2})=f_\text{greedy2},\hat{f}(X_\text{greedy1a})=f_\text{greedy1a}$. Since $X_\text{greedy1} \subseteq X_\text{greedy2}$, $f_\text{greedy2} \leq f_\text{greedy1}$. Note that in the cardinality constrained case we would have $X_1=X_2$, and hence $f_\text{greedy1}=f_\text{greedy2}$; therefore we use $f_\text{greedy}$ in that case. Then, as in \eqref{eq:bound} we know:
\begin{equation}\label{eq:boundbudget}
f_\text{convex}\leq f_\text{opt} \leq \min(f_\text{greedy2},f_\text{feas}) 
\end{equation} 

Again, for the case of $\hat{f}_\text{D}$, since $\hat{f}_\text{D}(X)-\hat{f}_\text{D}(\emptyset)$ is nonincreasing supermodular, we have two extra lower bounds for the value  $f_\text{opt}$ of \eqref{eq:optpmubudget}, which are derived from the proofs in \cite{krause2005note,khuller1999budgeted}: 
\begin{equation}\label{eq:boundsubmodbudget}\arraycolsep=1pt 
\begin{array}{c}
\tilde{f}_{\text{D},\text{greedy1}} \coloneqq (f_{\text{D},\text{greedy1}}-\hat{f}_\text{D}(\emptyset))\beta^{-1}+\hat{f}_\text{D}(\emptyset) \leq f_{\text{D},\text{opt}} \\[0.1cm]
\tilde{f}_{\text{D},\text{greedy1a}} \coloneqq (f_{\text{D},\text{greedy1a}}-\hat{f}_\text{D}(\emptyset))\beta_a^{-1}+\hat{f}_\text{D}(\emptyset) \leq f_{\text{D},\text{opt}} \\[0.1cm]
\end{array}
\end{equation}
with
\begin{equation}\arraycolsep=1pt 
\begin{array}{c}
\beta = \left( 1-\prod_{i\in X_\text{greedy1}} \left( 1-\frac{c_i}{b} \right)\right) \in (0,1] \\[0.1cm]
\beta_a = \left( 1-\prod_{i\in X_\text{greedy1a}} \left( 1-\frac{c_i}{b} \right)\right) \in (1-e^{-1},1]
\end{array}
\end{equation}

\begin{rem} 
If the costs $c_i$ are unitary and the budget $b$ equals the number of sensors, i.e $c_i=1 \; \forall i$ and $b=N_\text{meas}$, then $\alpha=\beta$. Consequently, \eqref{eq:boundsubmod} is a particular case of \eqref{eq:boundsubmodbudget}.
\end{rem}
A drawback of the bound with $\beta$ in \eqref{eq:boundsubmodbudget} is that $\beta$ does not have a lower bound bigger than $0$, and thus it may happen that $\beta^{-1} \to \infty$ and the bound becomes trivial. Nonetheless, we can prove the following result:

\begin{prop}\label{prop:beta}
Let $\gamma \in (0,1]$ define the percentage of budget used, so that $\sum_{i\in X_\text{greedy1}} c_i = \gamma b \leq b$, then we have
\begin{equation}\arraycolsep=1pt
\begin{array}{c}
\beta \in (1-e^{-\gamma},1], \beta_a \in (1-\gamma e^{-\gamma},1]
\end{array}
\end{equation}
\begin{proof}
See Appendix \ref{sec:appproofbeta}
\end{proof}
\end{prop}
Proposition \ref{prop:beta} shows that if $\gamma \to 0$, then $\beta \to 0$, but $\beta_a \to 1$; so only the bound with $\beta_a$ is useful. If $\gamma \to 1$, $\beta \to \beta_a$, and the bound with $\beta$ may be better.

Given the supermodularity of $f_\text{D}$, see Proposition \ref{prop:supermodular}, we can also consider the bound proposed in \cite{leskovec2007cost}, called online bound, generated by Algorithm \ref{alg:onlinebound} to obtain
\begin{equation}\label{eq:onlinebound}
\max_\mathcal{A}\tilde{f}_{\text{D},\text{online}}(\mathcal{A}) \leq f_{\text{D},\text{opt}}
\end{equation}

\begin{algorithm}[tbh]\label{alg:3}
\caption{Online bound}\label{alg:onlinebound}
\begin{algorithmic}[1]
\REQUIRE $b>0,c\in\mathbb{R}^n,c \geq 0,\mathcal{A}$ (any set)
\STATE $X \leftarrow \emptyset$
\STATE $\tilde{f}_{\text{D},\text{online}}=\hat{f}_{\text{D}}(\mathcal{A})$
\WHILE{$\sum_{i \in X} c_i < b$}
	\STATE $k \leftarrow \arg\min_{\{i\notin \mathcal{A}\}} \frac{\hat{f}_\text{D}(\mathcal{A}\cup \{i\})-\hat{f}_\text{D}(\mathcal{A})}{c_i}$
	\STATE $X \leftarrow X\cup \{k\}$
	\STATE $\tilde{f}_{\text{D},\text{online}} \leftarrow \tilde{f}_{\text{D},\text{online}}+\hat{f}_\text{D}(\mathcal{A}\cup \{k\})-\hat{f}_\text{D}(\mathcal{A})$
\ENDWHILE
\STATE $\tilde{f}_{\text{D},\text{online}} \leftarrow \tilde{f}_{\text{D},\text{online}}+(\hat{f}_\text{D}(\mathcal{A}\cup \{k\})-\hat{f}_\text{D}(\mathcal{A}))\frac{b-\sum_{i\in X} c_i}{c_k}$
\end{algorithmic}
\end{algorithm}

Finally we have all bounds:
\begin{equation}\label{eq:boundsubmodbudgetall}\begin{array}{c}
\max(\max_\mathcal{A}\tilde{f}_{\text{D},\text{online}}(\mathcal{A}),\tilde{f}_{\text{D},\text{greedy1}},\tilde{f}_{\text{D},\text{greedy1a}},f_{\text{D},\text{convex}}) \\[0.1cm]
\leq f_{\text{D},\text{opt}} \leq \min(f_\text{greedy2},f_\text{feas}) 
\end{array}
\end{equation}
\begin{rem}
Note that since $\mathcal{A}$ may be any set of sensors, $\max_\mathcal{A}\tilde{f}_{\mathrm{D},\text{online}}(\mathcal{A})$ is computationally demanding. 
\end{rem}

\vspace{-0.02cm}
\section{Test Case}\label{sec:testcase}
\vspace{-0.01cm}
Now we test the bounds on the 123-bus \cite{testfeeder} and the 8500-node \cite{testfeeder8500} test feeders. For the budget-constrained case in the 123-bus feeder, we consider two more expensive zones to represent heterogeneous costs: $\tilde{c}_i=2$ in the top left blue zone, $\tilde{c}_i=1.5$ in the bottom right red zone, $\tilde{c}_i=1$ elsewhere, see Fig. \ref{fig:123bus}. Then, sensor costs are normalized $c_i = \frac{\tilde{c}_i}{\sum_i \tilde{c}_i}$, so that their average is $1$, and thus the budget is approximately the number of sensors deployed. In the 8500-node feeder, we assign random normal distributed costs: $c_i \sim \mathcal{N}(1,0.1)$. For any number of measurements or budget, the sensor locations can be recovered from each solution $x$.

The algorithms are coded in Python and run on an Intel Core i7-6700HQ CPU at 2.60GHz with 16GB of RAM. For the convex optimization problems \eqref{eq:optpmuconvex} and \eqref{eq:optpmuconvexbudget}, we use the function \textit{minimize} of the \textit{scipy.optimize} package for the 123-bus test case, and a projected gradient descent algorithm for the 8500-node, see Appendix \ref{sec:appprojgrad}, where we use an efficient projection algorithm as in \cite{insense2018} in order to handle the computational complexity of the 8500-node feeder. Additionally, since function evaluations may take a few seconds for the 8500-node feeder, we use matrix algebra results to speed up the computation of the greedy solutions in \eqref{eq:optpmugreedy} and Algorithm \ref{alg:greedybudget}, see Appendix \ref{sec:appgreedy}. Otherwise, it could take several days to evaluate all possible nodes in every step.

\begin{figure}
\centering
\vspace{-0.2cm}
\includegraphics[width=7cm,height=4.7cm]{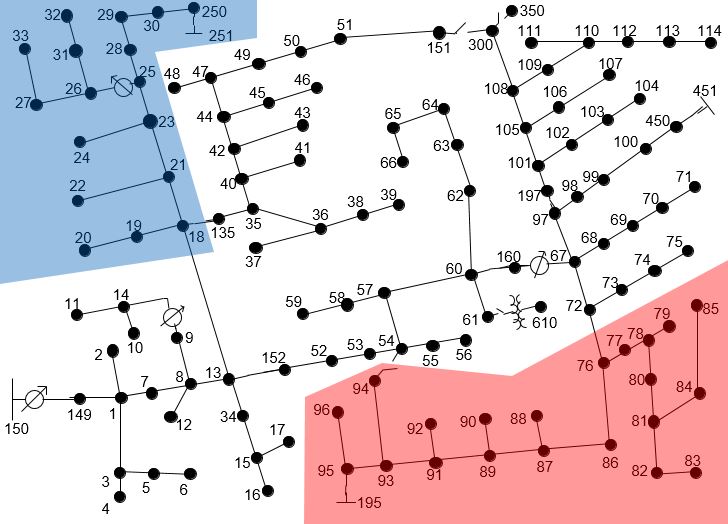}    
\vspace{-0.15cm}
\caption{123-bus test feeder from \cite{testfeeder} with different cost zones.} 
\vspace{-0.2cm}
\label{fig:123bus}
\end{figure}

Fig. \ref{fig:results} shows the bounds for the 123-bus and the 8500-node test feeders. Both the A,D-optimal metrics are analyzed, under a cardinality and a budget constraint, for different numbers of sensors and budgets respectively. For simplicity, in the online bound $\tilde{f}_{\text{D},\text{online}}(\mathcal{A})$ in \eqref{eq:onlinebound}, we have only considered $\mathcal{A}=\emptyset$ instead of any $\mathcal{A}$. However, we have observed that other options, like using a greedy solution for $\mathcal{A}$, produce similar bounds. The online bound has also been applied to the cardinality constrained problem. The yellow shaded area with squares shows the area between the minimum upper bound and the maximum lower bound, and thus the possible locations of the optimal solution $f_{\{\text{A,D}\},\text{opt}}$.

\begin{figure*}[!t]
\centering
\subfloat[123-bus A-optimal with cardinality constraint]{\includegraphics[width=8cm]{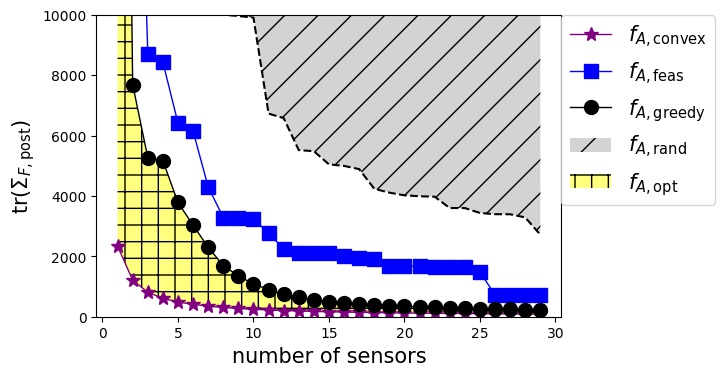}
\label{fig:tracesigma}}
\subfloat[123-bus A-optimal with budget constraint]
{\includegraphics[width=8cm]{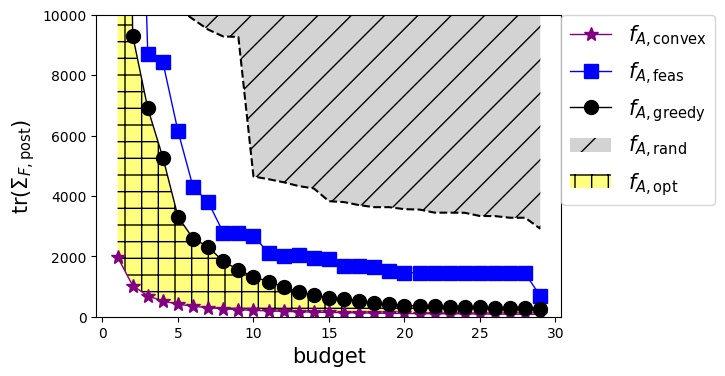}
\label{fig:tracesigmabudget}}
\hfil
\vspace{-0.15cm}
\subfloat[123-bus D-optimal with cardinality constraint]{\includegraphics[width=8cm]{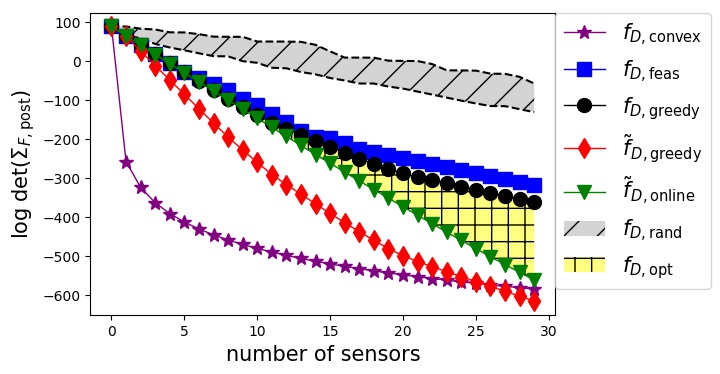}
\label{fig:logdetsigma}}
\subfloat[123-bus D-optimal with budget constraint]
{\includegraphics[width=8cm]{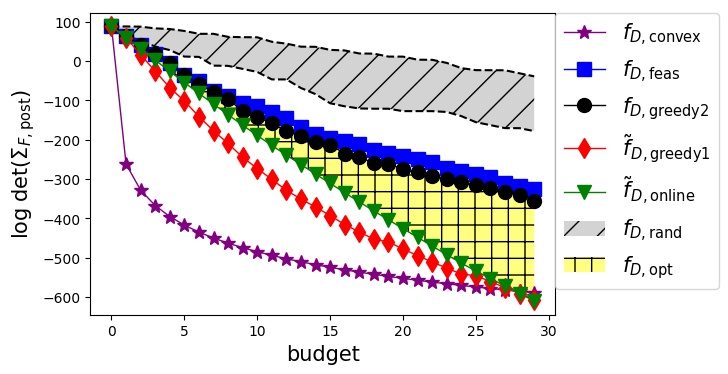}
\label{fig:logdetsigmabudget}}
\hfil
\vspace{-0.00cm}
\subfloat[8500-node A-optimal with cardinality constraint]{\includegraphics[width=8cm]{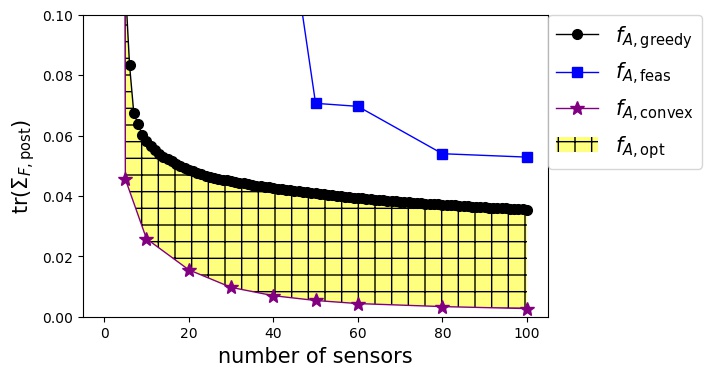}
\label{fig:tracesigmabig}}
\hspace{-0.15cm}
\subfloat[8500-node A-optimal with budget constraint]{\includegraphics[width=8cm]{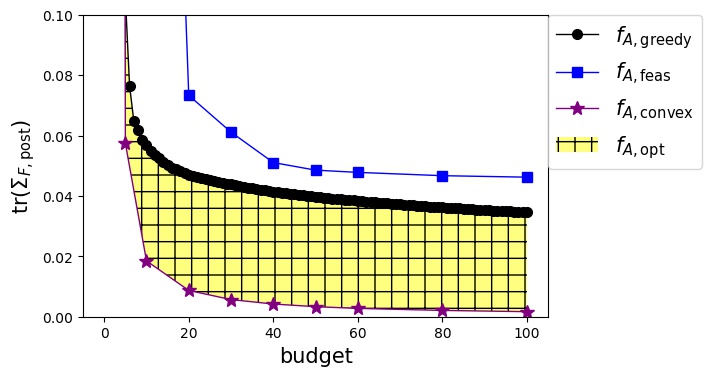}
\label{fig:tracesigmabudgetbig}}
\vspace{-0.00cm}
\subfloat[8500-node D-optimal with cardinality constraint]{\includegraphics[width=8cm]{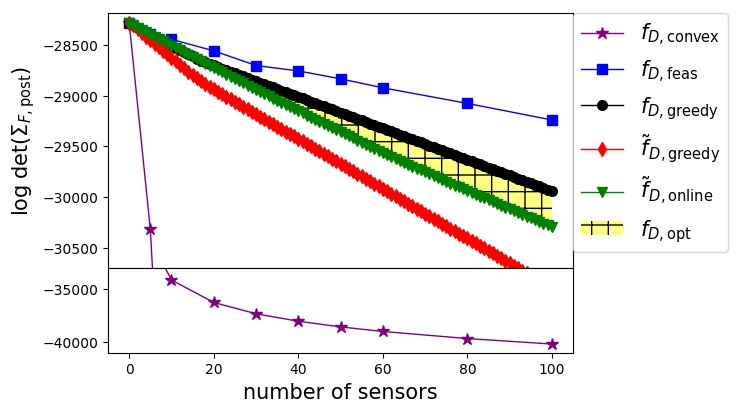}
\label{fig:logdetsigmabig}}
\hspace{-0.15cm}
\subfloat[8500-node D-optimal with budget constraint]{\includegraphics[width=8cm]{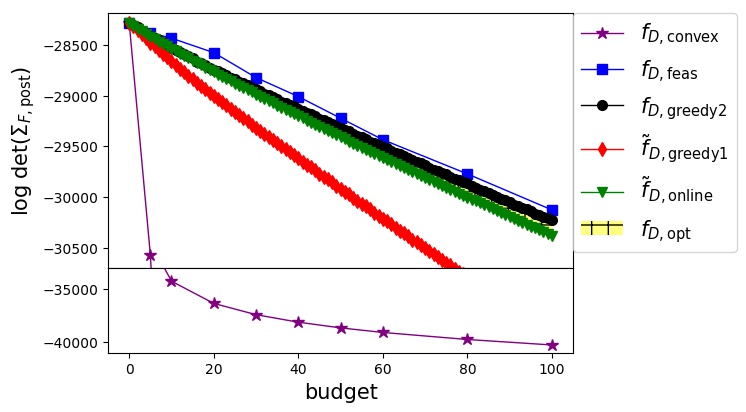}
\label{fig:logdetsigmabudgetbig}}
\vspace{-0.1cm}
\caption{Plots for the A,D-optimal metrics, under cardinality and budget constraints, the 123-bus and 8500-node cases, showing the lower bounds and the upper bounds. The grey shaded area filled with oblique lines shows the values $f_{\{\text{A,D}\},\text{rand}}$ for random configurations of sensors (100 samples). The yellow shaded area filled with a grid of horizontal and vertical lines shows the possible locations of the optimal solution $f_{\{\text{A,D}\},\text{opt}}$.}
\label{fig:results}
\end{figure*}

For the D-optimal metric, in Fig. \ref{fig:logdetsigma}, \ref{fig:logdetsigmabudget}, \ref{fig:logdetsigmabig} and \ref{fig:logdetsigmabudgetbig}, it can be observed that the online bound $\tilde{f}_{\text{D},\text{online}}$ \eqref{eq:onlinebound} outperforms the rest of the lower bounds ($\tilde{f}_{\text{D},\text{greedy}}$ \eqref{eq:boundsubmod}, $\tilde{f}_{\text{D},\text{greedy1}}$ \eqref{eq:boundsubmodbudget}, $f_{\text{D},\text{convex}}$ \eqref{eq:optpmuconvex}, \eqref{eq:optpmuconvexbudget}) for a small number of sensors under a cardinality and under a budget constraint, especially in the 8500-node feeder. For clarity, $\tilde{f}_{\text{D},\text{greedy1a}}$ \eqref{eq:boundsubmodbudget} has not been plotted, since the results were almost equal to $\tilde{f}_{\text{D},\text{greedy1}}$. 
For a small number of sensors, it is remarkable how the lower bound based on convex relaxations $f_{\text{D},\text{convex}}$ performs poorly compared to the other lower bounds, as it produces too optimistic results. This is a result of the high SE uncertainty in all nodes of the grid before the PMU placement; as a consequence, the convex optimization allocates partial unfeasible PMUs with $0 < x_i\ll 1$ in many nodes of the grid, instead of full PMUs with $x_i=1$ in a few nodes. This is equivalent to adding more sensors with an error with larger relative standard deviation. This cannot be done, since PMUs have a fixed maximum relative error by design \cite{martin2008exploring}, see section \ref{sec:info}.

For the A-optimal metric, in Fig. \ref{fig:tracesigma}, \ref{fig:tracesigmabudget}, \ref{fig:tracesigmabig} and \ref{fig:tracesigmabudgetbig} we do not have the lower bounds based on the supermodularity property; however, the lower bound based on convex relaxations $f_{\text{A},\text{convex}}$ \eqref{eq:optpmuconvex}, \eqref{eq:optpmuconvexbudget} performs better than for the D-optimal metric. It still can be observed that for a small number of sensors this convex bound is far from the greedy and feasible solutions ($f_{\text{A},\text{feas}}$ \eqref{eq:optpmufeas}, $f_{\text{A},\text{greedy}}$ \eqref{eq:optpmugreedy}, $f_{\text{A},\text{greedy2}}$ in Algorithm \eqref{alg:greedybudget}), but not as much as in the D-optimal metric case, and the bound approaches the values of the solutions quicker.

As expected, for both metrics simple greedy or feasible solutions ($f_{\{\text{A,D}\},\text{greedy}}$, $f_{\{\text{A,D}\},\text{greedy2}}$, $f_{\{\text{A,D}\},\text{feas}}$) perform better than random configurations ($f_{\{\text{A,D}\},\text{rand}}$), even better than the best random configuration out of the 100 samples displayed.

\vspace{-0.02cm}
\section{Conclusions}\label{sec:conc}
\vspace{-0.01cm}
We have stated the problem of optimal sensor placement for deploying PMUs to minimize the uncertainty of state estimation in distribution grids, both under a cardinality and a budget constraint. We have analyzed the properties of different metrics, concretely, convexity and supermodularity. Using these properties, we have derived a set of bounds that have enabled us to narrow the gap of the possible value of the optimal solution of this intractable problem. Since the optimal solution is unpractical to obtain in large grids, suboptimal solutions are required instead. Using these bounds, we can compute the maximum gap between any suboptimal solution and the optimal solutions. Moreover, we have observed, that the bounds produced by the supermodularity property are especially relevant when only a small number of sensor can be installed in large networks. In this case, a bound based on convex relaxation produces a too optimistic result.

Future work could include extending these results to take into account network reconfiguration due to different states of switches. Moreover, more exhaustive search algorithms could be developed to obtain solutions closer to the lower bounds. Additionally, it would be interesting to analyze how the choice of the metric affects the effectiveness of the bounds.

\vspace{-0.02cm}
\appendices
\vspace{-0.01cm}
\section{Proof of Proposition \ref{prop:supermodular}}\label{sec:appsupermodD}
\vspace{-0.01cm}
\begin{proof}
The constant term $\hat{f}_\text{D}(\emptyset)$ is only necessary to ensure that the function is $0$ when $X=\emptyset$. For the rest of the proof this term is not necessary since for any two sets $X,Y$ we have: $(\hat{f}_\text{D}(Y)-\hat{f}_\text{D}(\emptyset))-(\hat{f}_\text{D}(X)-\hat{f}_\text{D}(\emptyset))=\hat{f}_\text{D}(Y)-\hat{f}_\text{D}(X)$. \\
This proof is very similar to the one used for controlability gramians in \cite{summers2016submodularity}. Consider the set of sensors $X\subseteq Y$ and a sensor $a$ so that $\{a\}\notin Y$. Now let us define
\begin{equation*}\arraycolsep=1pt \begin{array}{rl}
g_a(X)\coloneqq &\hat{f}_\text{D}(X\cup \{a\})-\hat{f}_\text{D}(X) \\[0.1cm]
\tilde{f}_\text{D}(\gamma,X,Y)\coloneqq &\log\text{det}\big((\Sigma_{F,\text{prior}}^{-1}+ \Lambda_X+\gamma \Lambda_{Y\setminus X})^{-1}\big) \\[0.1cm] 
\tilde{g}_a(\gamma,X,Y)\coloneqq & \tilde{f}_\text{D}(\gamma,X\cup \{a\},Y\cup \{a\})-\tilde{f}_\text{D}(\gamma,X,Y)
\end{array}
\end{equation*}
So that we have $\tilde{f}_\text{D}(0,X,Y)=\hat{f}_\text{D}(X)$, $\tilde{f}_\text{D}(1,X,Y)=\hat{f}_\text{D}(Y)$, $\tilde{g}_a(0,X,Y)=g_a(X)$ and $\tilde{g}_a(1,X,Y)=g_a(Y)$. Now we will prove that $g_a(Y)\geq g_a(X)$ by computing the gradient of $\tilde{g}_a$:
\begin{equation*}\arraycolsep=1pt \begin{array}{rl}
\frac{\partial \tilde{f}_\text{D}}{\partial \gamma}(\gamma,X,Y)
 = & -\text{tr}\big((\Sigma_{F,\text{prior}}^{-1}+ 
\Lambda_X+\gamma \Lambda_{Y\setminus X})^{-1}\Lambda_{Y\setminus X}\big) \\[0.1cm]

\frac{\partial \tilde{g}_a}{\partial \gamma}(\gamma,X,Y) 

= &\frac{\partial \tilde{f}_\text{D}}{\partial \gamma}(\gamma,X\cup \{a\},Y\cup \{a\}) - \frac{\partial \tilde{f}_\text{D}}  {\partial \gamma}(\gamma,X,Y) \\[0.1cm]

= &-\text{tr}\Big(\big((\Sigma_{F,\text{prior}}^{-1}+ 
\Lambda_{X\cup \{a\}}+\gamma \Lambda_{Y\setminus X})^{-1} \\[0.1cm]
&-(\Sigma_{F,\text{prior}}^{-1}+ \Lambda_X+\gamma \Lambda_{Y\setminus X})^{-1}\big)
\Lambda_{Y\setminus X}\Big)
\end{array}
\end{equation*}

Since $X\subseteq Y$, we have $\Lambda_{Y\setminus X}\succeq 0$. Since $\Lambda_{\{a\}}\succeq 0$ and $\Lambda_{X\cup \{a\}} = \Lambda_{\{a\}} +\Lambda_X$, we have
\begin{equation*}
\Sigma_{F,\text{prior}}^{-1} + \Lambda_{X\cup \{a\}}+\gamma \Lambda_{Y\setminus X}\succeq\Sigma_{F,\text{prior}}^{-1}+ 
\Lambda_{X}+\gamma \Lambda_{Y\setminus X}\succeq 0
\end{equation*}
and thus 
\begin{equation*}
(\Sigma_{F,\text{prior}}^{-1}+ \Lambda_{X}+\gamma \Lambda_{Y\setminus X})^{-1}\succeq(\Sigma_{F,\text{prior}}^{-1}+ 
\Lambda_{X\cup \{a\}}+\gamma \Lambda_{Y\setminus X})^{-1}
\end{equation*}

Consequently, $\frac{\partial \tilde{g}_a}{\partial \gamma}(\gamma,X,Y) \geq 0$ and
\begin{equation*} \begin{array}{c}
g_a(Y)-g_a(X)=\tilde{g}_a(1,X,Y)-\tilde{g}_a(0,X,Y) \\[0.1cm]
=\int_0^1 \frac{\partial \tilde{g}_a}{\partial \gamma}(\gamma,X,Y) \partial \gamma \geq 0
\end{array}
\end{equation*}
so that
\begin{equation*}
\hat{f}_\text{D}(Y\cup \{a\})-\hat{f}_\text{D}(Y) \geq \hat{f}_\text{D}(X\cup \{a\})-\hat{f}_\text{D}(X)
\end{equation*}
\vspace{-0.7cm}
\end{proof}

\vspace{-0.02cm}
\section{Proof of Proposition \ref{prop:beta}}\label{sec:appproofbeta}
\vspace{-0.01cm}
\begin{proof}
If $\sum_{i\in X_\text{greedy1}} c_i = \gamma b$, then $\prod_{i\in X_\text{greedy1}} \left( 1-\frac{c_i}{b} \right)$ achieves its maximum at $c_i = \frac{\gamma b}{\abs{X_\text{greedy1}}}$ for all $i$, where $|\cdot|$ is the number of elements in a set. Then we have: 
\begin{equation*}\arraycolsep=1pt
\begin{array}{rl}
\beta & = ( 1-\prod_{i\in X_\text{greedy1}} ( 1-\frac{c_i}{b})) \\[0.1cm]
&\geq \Big( 1- \Big( 1-\frac{\gamma}{\abs{X_\text{greedy1}}} \Big)^{\abs{X_\text{greedy1}}} \Big) \geq ( 1- e^{-\gamma} )
\end{array}
\end{equation*}
We assume that $c_i \leq b$ for all $i$ (sensors with $c_i > b$ are discarded since they cannot be installed). Then we know that $b>c_a > b - \gamma b$, where $c_a$ is the cost of element $a$, and we have:
\begin{equation*}\arraycolsep=1pt
\begin{array}{rl}
\beta_a & = ( 1-\prod_{i\in X_\text{greedy1a}} ( 1-\frac{c_i}{b} )) \\[0.1cm]
&\geq \Big( 1- \Big( 1-\frac{\gamma}{\abs{X_\text{greedy1}}} \Big)^{\abs{X_\text{greedy1}}} (1-\frac{c_a}{b}) \Big) \\[0.1cm]
& \geq ( 1-e^{-\gamma}(1-\frac{b - \gamma b}{b}) ) = (1-\gamma e^{-\gamma})
\end{array}
\end{equation*}
\vspace{-0.7cm}
\end{proof}

\vspace{-0.02cm}
\section{Efficient Projected Gradient Descent}\label{sec:appprojgrad}
\vspace{-0.01cm}
To solve the convex problem under a cardinality constraint \eqref{eq:optpmuconvex} for the 8500-node feeder, we use a projected gradient descent method using the gradient expressions in \eqref{eq:grad}:
\begin{equation*}\begin{array}{l}
x^{(K+1)} = \Pi_\mathcal{X}(x^{(K)}-\alpha^{(K)}\nabla f(x^{(K)})) \\
\mathcal{X} = \{ x \mid \sum_i x_i = N_\text{meas},\; x_i \in [0,1] \}
\end{array}
\end{equation*}
where the suffix $(\cdot)^{(K)}$ denotes the value at iteration $K$. We use $\alpha^{(K)} = \frac{\alpha}{k\norm{\nabla f(x^{(K)})}}_2$ to guarantee convergence of the method \cite{nedic2001incremental}, where $\alpha$ is a design parameter. For a more efficient projection $\Pi_\mathcal{X}(\cdot)$, we use the scaled boxed-simplex projection algorithm proposed in \cite{insense2018}, which converges in a finite number of iterations, as opposed to the projection algorithm proposed in \cite{kekatos2012optimal}.

When solving the problem under a budget constraint \eqref{eq:optpmuconvexbudget}, we can use the change of variables $y_i = x_i c_i$, and the function $f_c(y)=f\Big(\frac{y}{c}\Big)$, and solve the alternative problem
\begin{equation*}
y_\text{convex} = \arg\min_y f_c(y) \; s.t. \; \sum_i y_i \leq b,\; y_i\in  [0,c_i]\; \forall i
\end{equation*}
using the projected gradient descent method:
\begin{equation*}\begin{array}{l}
y^{(K+1)} = \Pi_\mathcal{Y}(y^{(K)}-\alpha^{(K)}\nabla f_c(y^{(K)})) \\
\mathcal{Y} = \{ y \mid \sum_i y_i = b,\; y_i \in [0,c_i] \}
\end{array}
\end{equation*}
where an efficient implementation of the projection $\Pi_\mathcal{Y}(\cdot)$ can be derived by using a modified version of the one in \cite{insense2018} by changing the $1$ to $c_i$ for each respective $y_i$. The details of the algorithm can be found in \cite{picallo2018convexPMU}. 

\vspace{-0.02cm}
\section{Efficient Greedy Computations}\label{sec:appgreedy}
\vspace{-0.01cm}
The posterior covariance at iteration $K$ can be defined as
\begin{equation*}
\Sigma_{F,\text{post}}^{(K)}
=(\Sigma_{F,\text{prior}}^{-1}+\sum_{i \in X^{(K)}}(C_\text{meas}F)_{i,\bullet}^*(C_\text{meas}F)_{i,\bullet}(\Sigma_\text{meas}^{-1})_{i,i})^{-1}
\end{equation*}
With this expression, we use Woodbury's matrix identity \cite{matrixcalc} to obtain a simplified expression that is easier to evaluate:
\begin{equation*} \begin{array}{l}
\tilde{f}_\text{A}(X^{(K)} \cup \{j\}) \\
= \text{tr} \big( ( (\Sigma_{F,\text{post}}^{(K)})^{-1}+(C_\text{meas}F)_{j,\bullet}^*(C_\text{meas}F)_{j,\bullet}(\Sigma_\text{meas}^{-1})_{j,j} )^{-1} \big) \\
= \text{tr} \big( \Sigma_{F,\text{post}}^{(K)} - 
\frac{\Sigma_{F,\text{post}}^{(K)}(C_\text{meas}F)_{j,\bullet}^*
(C_\text{meas}F)_{j,\bullet}\Sigma_{F,\text{post}}^{(K)}}{(\Sigma_\text{meas}^{-1})_{j,j}^{-1} + (C_\text{meas}F)_{j,\bullet}\Sigma_{F,\text{post}}^{(K)} (C_\text{meas}F)_{j,\bullet}^*} \big) \\
= \text{tr}(\Sigma_{F,\text{post}}^{(K)}) - \frac{(C_\text{meas}F)_{j,\bullet}(\Sigma_{F,\text{post}}^{(K)})^2(C_\text{meas}F)_{j,\bullet}^*}{(\Sigma_\text{meas}^{-1})_{j,j}^{-1} + (C_\text{meas}F)_{j,\bullet}\Sigma_{F,\text{post}}^{(K)} (C_\text{meas}F)_{j,\bullet}^*}
\end{array}
\end{equation*}
Then, we use Sylvester's determinant identity \cite{matrixcalc} to get
\begin{equation*} \begin{array}{l}
\tilde{f}_\text{D}(X^{(K)} \cup \{j\}) \\
= \log\text{det} \big( ( (\Sigma_{F,\text{post}}^{(K)})^{-1} \hspace{-0.15cm} + \hspace{-0.1cm} (C_\text{meas}F)_{j,\bullet}^*(C_\text{meas}F)_{j,\bullet}(\Sigma_\text{meas}^{-1})_{j,j} )^{-1} \hspace{-0.05cm} \big) \\
= \log\text{det} (\Sigma_{F,\text{post}}^{(K)}) \\
- \log\text{det}\big( I_\text{d} + \Sigma_{F,\text{post}}^{(K)}(C_\text{meas}F)_{j,\bullet}^*(C_\text{meas}F)_{j,\bullet}(\Sigma_\text{meas}^{-1})_{j,j} \big) \\
= \log\text{det} (\Sigma_{F,\text{post}}^{(K)}) \\
- \log(1+(\Sigma_\text{meas}^{-1})_{j,j}(C_\text{meas}F)_{j,\bullet}\Sigma_{F,\text{post}}^{(K)}(C_\text{meas}F)_{j,\bullet}^*)
\end{array}
\end{equation*}

\bibliographystyle{IEEEtran}
\bibliography{ifacconf}



%

\begin{IEEEbiographynophoto}{Miguel Picallo}
received two Diplomas in Mathematics and Industrial Engineering and the M.Sc. in Management Science and Engineering from Stanford University. He is currently pursuing his PhD about state estimation and optimal power flow in Distribution Grids at the Delft University of Technology, and collaborated with GE Global Research, Germany.
\end{IEEEbiographynophoto}

\vskip 0pt plus -1fil

\begin{IEEEbiographynophoto}{Adolfo Anta}
is currently a senior researcher at the Austrian Institute of Technology (AIT), Vienna, Austria. He received the M.Sc. and Ph.D. degrees in Control Systems from the University of California at Los Angeles, USA, in 2007 and 2010, respectively. From 2012 to 2018 he worked as lead researcher at GE Global Research Europe, Germany. His research interests cover a wide range of control applications, in particular stability issues in power systems.
\end{IEEEbiographynophoto}

\vskip 0pt plus -1fil

\begin{IEEEbiographynophoto}{Bart De Schutter}
(IEEE member since 2008, senior member since 2010, fellow since 2019) is a full professor at the Delft Center for Systems and Control of Delft University of Technology in Delft, The Netherlands. He is senior editor of the IEEE Transactions on Intelligent Transportation Systems. His current research interests include intelligent transportation, infrastructure  and power systems, hybrid systems, and multi-level control.
\end{IEEEbiographynophoto}

\end{document}